\documentclass[oneside,reqno]{amsart}
\usepackage{amssymb,amsmath}
\usepackage[english]{babel}
\usepackage{amsfonts}
\usepackage{color}
\usepackage{amsthm}
\usepackage[colorlinks=true,linkcolor=NavyBlue, citecolor=NavyBlue,urlcolor=NavyBlue]{hyperref}
\usepackage{graphicx}
\usepackage{mathtools}
\usepackage[usenames,dvipsnames]{pstricks}
\usepackage{bm}
\usepackage{amsmath}
\usepackage{listings}
\usepackage{hyperref}

\usepackage{xy}
\usepackage[draft]{fixme}

\DeclareMathOperator{\AGL}{AGL}
\DeclareMathOperator{\GL}{GL}
\DeclareMathOperator{\Sym}{Sym}

\DeclareMathOperator{\F}{\mathbb{F}}

\newcommand\deq{\mathrel{\stackrel{\makebox[0pt]{\mbox{\normalfont\tiny def}}}{=}}}

\newcommand{\Span}[1]{\left\langle#1\right\rangle}

\let\phi\varphi
\theoremstyle{plain}
\newtheorem{theorem}{Theorem}[section]

\newtheorem{lemma}[theorem]{Lemma}

\newtheorem{corollary}[theorem]{Corollary}

\theoremstyle{remark}

\theoremstyle{definition}
\newtheorem{definition}[theorem]{Definition}

\newtheorem*{notation*}{Notation}

\usepackage{listings}
\lstdefinelanguage{GAP}{%
 morekeywords={%
 Assert,Info,IsBound,QUIT,%
 TryNextMethod,Unbind,and,break,%
 continue,do,elif,%
 else,end,false,fi,for,%
 function,if,in,local,%
 mod,not,od,or,%
 quit,rec,repeat,return,%
 then,true,until,while%
 },%
 sensitive,%
 morecomment=[l]\#,%
 morestring=[b]",%
 morestring=[b]',%
}[keywords,comments,strings]

\usepackage[T1]{fontenc}
\usepackage[variablett]{lmodern}
\usepackage{xcolor}
\DeclareMathAlphabet{\mbb}{U}{bbold}{m}{n}
\let \varphi \phi
\let \theta \pi
\newcommand{\veps}{\varepsilon}

\newcommand{\rb}{\overline{\rho}}
\newcommand{\tb}{\overline\theta}
\newcommand{\M}{M}
\newcommand{\rhoLM}{
\begin{pmatrix}
\mbb 1 & \mbb1\\
\mbb1 & \mbb0
\end{pmatrix} 
\begin{pmatrix}
\mbb 1& \mbb1+\rho\\
\mbb 0 & \mbb1
\end{pmatrix}
}
\newcommand{\rhoinvLM}{
\begin{pmatrix}
\mbb1+\rho & \rho\\
\mbb1 & \mbb1
\end{pmatrix}
}
\newcommand{\thetainvLM}
{
\begin{pmatrix}
\theta^{-1} & \mbb0\\
\mbb1 & \mbb1
\end{pmatrix} 
}

\begin{document}
\title[On the primitivity of Lai-Massey schemes]{On the primitivity of  Lai-Massey schemes} 
 \author[R.~Aragona]{Riccardo Aragona}
\author[R.~Civino]{Roberto Civino}

\address{DISIM \\
 Universit\`a degli Studi dell'Aquila\\
 via Vetoio\\
 67100 Coppito (AQ)\\
 Italy}       

\email{riccardo.aragona@univaq.it}
\email{roberto.civino@univaq.it} 

\date{} \thanks{All the authors are members of INdAM-GNSAGA
 (Italy). This work was partially supported by the Centre of EXcellence on Connected, Geo-Localized and 
 Cybersecure Vehicles (EX-Emerge), funded by Italian Government under CIPE resolution n. 70/2017 (Aug. 7, 2017).}


\begin{abstract}
In symmetric cryptography, the round functions used as building blocks for iterated block ciphers
are often obtained as the composition of different layers providing confusion and diffusion.
The study of the conditions on such layers which make the group generated by the round functions of a block cipher
 a primitive group  has been addressed in the past 
years, both in the case of Substitution Permutation Networks and Feistel Networks, 
giving to block cipher designers the receipt to avoid the imprimitivity attack. In this paper a similar study is 
proposed on the subject of the Lai-Massey scheme, a framework which combines both Substitution Permutation Network
and Feistel Network features. 
Its resistance to the imprimitivity attack is obtained as a consequence of a more general
result in which the problem of proving the primitivity of the Lai-Massey scheme is reduced to the simpler one of proving 
the primitivity of
the group generated by the round functions of a strictly related Substitution Permutation Network.
\end{abstract}


\maketitle


\section{Introduction}
Until the selection of the Advanced Encryption Standard~\cite{aes}, Feistel Networks (FN) have probably been the most popular 
design framework for iterated block ciphers, whereas today they share the stage with Substitution Permutation Networks (SPN).
Feistel Networks are characterized by the clever idea of splitting the message into two halves, say left part and right part, and applying in each round a key-dependent non-linear transformation called F-function
to the right part, which is successively mixed with the left part, just before the two halves are swapped~\cite{feistel}.
As a notable feature, FNs do not require the F-function to be invertible in order to perform decryption. The framework of SPNs is instead composed 
by a sequence of carefully designed key-dependent round functions composed by confusion and diffusion invertible layers acting on the whole block.
If, on the one hand, SPNs' minimalistic design allows a simple description and consequently a more careful security assessment, 
on the other, the structure of FNs gives the designers more freedom in the choice of the layers intervening during the encryption, 
although keeping confusion and diffusion confined only in half of the block in a single round. The Lai-Massey scheme (LM)~\cite{vaudenay1999lai}, introduced after the design of IDEA~\cite{lai1990proposal}, perfectly 
combines the advantages of both frameworks, splitting the message into two halves but mixing the left and right part of the state and 
consequently accelerating both diffusion and confusion.
Its pseudo-randomness behavior, its security against impossible differential cryptanalysis and other generic attacks has been addressed in recent years~\cite{quasifeistel,impossibledifferentialLM,randomlaimassey,genericattackLM}.\\

In this paper we focus on the study of a group containing the group generated by the round functions of a general Lai-Massey cipher, proving its resistance against the \emph{imprimitivity 
attack}~\cite{paterson1999imprimitive} provided that its inner layers satisfy certain well-established conditions. 
Little is known, indeed, on the group-theoretical security of such a design strategy, whereas the one of SPNs and FNs has been addressed in several works in the last decades. One remarkable  exception is a paper due to Wernsdorf~\cite{wenidea} which shows that the multiply-addition box at the center of the round of IDEA generates the alternating group on $\F_2^{32}$ and where it is conjectured that also the entire rounds of IDEA generate the alternating group.

The topic of our research, i.e.\ the group generated by the round functions, was first defined in 1975 by Coppersmith and Grossman~\cite{copgross} and gained more popularity
  when, in 1999, Paterson introduced the imprimitivity attack showing that in a DES-like cipher may exist a partition of the message space which is invariant under the action of the group $\Gamma$ generated by the encryption functions, i.e.\ a block system for $\Gamma$, whose knowledge can be exploited to attack the cipher. After this, the resistance of many known ciphers 
to this attack has been proved~\cite{sparwenrij,carantiprimitive,wenkas,aragona2017group,aragona2018primitivity}.
In Aragona et al.~\cite[Theorem 4.5]{aragonawave}, the authors showed that the primitivity of the group generated by the rounds of an FN can be \emph{reduced} to the primitivity of the group generated by the rounds of an SPN whose round functions are the ones implemented as F-functions within each round of the FN, proving, in fact, that the primitivity of structure of an FN, in spite of its complexity, 
can be inherited from a simpler design.
We prove here, using a similar approach, that the primitivity of the group generated by the rounds of an SPN implies the one of a group containing the group generated by the rounds of an LM which features in its structure the same key-dependent transformation acting in the SPN. Our result is referred to the closest group containing the actual group generated by the round functions of an LM for which a convenient algebraic description of the generators can be provided.
\subsubsection*{Organization of the paper}
In Section~\ref{sec:pre} we introduce the notation and the preliminary results, and  present our algebraic model 
 of Lai-Massey scheme which is the subject of the study. In Section~\ref{sec:pri} we prove the primitivity 
 reduction from the LM to the SPN.


\section{Group-theoretical cryptanalysis}\label{sec:pre}
\subsection{Preliminaries}
Let us introduce our notation and some preliminary results.
\subsubsection{Spaces}
Let $n$ be a non-negative integer and $V \deq \F_2^n$ be the $n$-dimensional vector space over $\F_2$. We denote by $\Sym(V)$ the symmetric group acting on $V$ and by $\mbb1$ its identity. The map $\mbb 0 : V \rightarrow V$ denotes the null function. The group of the translations on V, i.e.\ the group of the maps $\sigma_v: V \rightarrow V$, such that $x\sigma_v= x+v$, is denoted by $T_n$, whereas the group of translations on $V \times V$ is denoted by $T_{2n}$, where the translation $\sigma_{(v,w)}$ acts on $(x,y)$ as 
$(x,y)\sigma_{(v,w)} = (x+v, y+w)$. Let us also denote by $\AGL(V)$ the group of all affine permutations of $V$ and by $\GL(V)$ the group of the linear ones.
\subsubsection{Groups}
Let $G$ be a finite group acting on a set $\M$. For each $g \in G$ and $v \in \M$ we denote the action of $g$ on $v$ as $vg$. 
The group $G$ is said to be \emph{transitive} on $\M$ if for each $v,w \in \M$ there exists $g \in G$ such that $vg=w$.
A partition $\mathcal{B}$ of $\M$ is \emph{trivial} if $\mathcal{B}=\{\M\}$ or $\mathcal{B}=\{\{v\} \mid v \in \M\}$, and \emph{$G$-invariant} if for any $B \in \mathcal{B}$ and $g \in G$ it holds $Bg \in \mathcal{B}$. Any non-trivial and $G$-invariant partition $\mathcal{B}$ of $\M$ is called a \emph{block system} for $G$. In particular any $B \in \mathcal{B}$ is called an \emph{imprimitivity block}. The group $G$ is \emph{primitive} in its action on $\M$ (or $G$ \emph{acts primitively} on $\M$) if {$G$ is transitive and} there exists no block system. 
Otherwise,  the group $G$ is \emph{imprimitive} in its action on $\M$ (or $G$ \emph{acts imprimitively} on $\M$). We recall here some well-known results that will be useful in the remainder of this paper~\cite{cameronpermutation}. 

\begin{lemma}\label{lemma:trans}
  If $T \leq G$ is transitive, then a block system
  for $G$ is also a block system for $T$.
\end{lemma}

\begin{lemma}\label{translatioBlocks}
Let $\M$ be a finite vector space over $\mathbb F_2$ and $T$ its translation group. 
Then $T$ is transitive and imprimitive on $\M$. A block system $\mathcal U$ for $T$ is composed by
  the cosets  of a non-trivial and proper subgroup $ U < (M,+)$, i.e. 
  \begin{equation*}
 \mathcal U =    \{ 
     U + v
      \mid
 	v \in    M
      \}.
        \end{equation*}
  \end{lemma}
  \subsubsection{Goursat's Lemma}
To prove our results, we need to determine a block system for $V \times V$.
In order to do so, we use the following 
characterization of subgroups of the direct product  of two groups in terms of
suitable sections of the direct factors~\cite{goursat}.  

\begin{theorem}\label{gours}
  Let $G_1$  and $G_2$ be two  groups. There
  exists  a bijection between  
  \begin{enumerate}
  \item 
    the set  of all subgroups  of the 
    direct  product  $G_1\times   G_2$,  and  
  \item 
    the  set   of  all  triples
    $(A/B,C/D,\psi )$, where 
    \begin{itemize}
    \item $A$ is a subgroup of $G_{1}$;
    \item $C$ is a subgroup of $G_{2}$;
    \item $B$ is a normal subgroup of $A$;
    \item $D$ is a normal subgroup of $C$; 
    \item $\psi: A/B\to C/D$ is a group isomorphism.
    \end{itemize}
  \end{enumerate}
  
\noindent Then, each subgroup of $G_1\times G_2$ can be uniquely
  written as
  \begin{equation*}
    U_{\psi}= \{
      (a,c) \in A \times C 
      :
      (a + B) \psi =c + D
      \}.
  \end{equation*}
\end{theorem}

Note that the isomorphism $\psi: A/B\to C/D$ is induced by a homomorphism $\varphi: A \to C$ such that $(a+B)\psi=a\varphi + D$ for any $a\in A$, and $B\varphi\leq D$. Such homomorphism is not unique. 
\begin{corollary}[\cite{aragonawave}]\label{lemma:psiforphi}
  Using notation of Theorem~\ref{gours}, given any homomorphism $\phi$ inducing $\psi$, we have 
 \begin{equation}\label{eq:upsi}
    U_{\psi}
    =
    \{
      (a, a \varphi + d)
      \mid
      a \in A, d \in D
      \}.
  \end{equation}
\end{corollary}
\subsubsection{Ciphers}
A \emph{block cipher} $\Phi$ is a family of key-dependent permutations 
\[\{E_K \mid E_K: \ M \rightarrow \M, \, K \in \mathcal K \},\]
 where $\M$ is the message space, $\mathcal K$ the key space, and $|\M|\leq |\mathcal K|$. The permutation $E_K$ is called the \emph{encryption function induced by the master key} $K$. The block cipher $\Phi$ is called an {iterated block cipher} if there exists $r \in \mathbb N$ such that for each $K \in \mathcal K$ the encryption function $E_K$ is obtained as the composition of $r$ {round functions}, i.e. $E_K = \veps_{1,K}\,\veps_{2,K}\ldots\veps_{r,K}$. To provide efficiency, each round function is the composition of a public component provided by the designers, and a private component derived from the user-provided key by means of a public procedure known as \emph{key-schedule}. The group
\begin{equation*}\label{def:gr}
	\Gamma_{\infty}(\Phi)\deq\langle \veps_{i,K} \mid K \in \mathcal{K}, 1 \leq i \leq r\rangle,
\end{equation*}
called \emph{the group generated by the round functions} of $\Phi$, is studied to prevent 
group-theoretical attacks~\cite{kaliski,paterson1999imprimitive,calderini2017translation}.\\

An iterated {block cipher} $\Phi$ is called an $r$-round \emph{Substitution Permutation Network} (SPN) if $M = V$ and for each $1 \leq i \leq r$ we have 
 \[
 \veps_{i,K} \deq \rho\sigma_{k_i},
 \]
 where $\rho\in \Sym(V)\setminus \AGL(V)$ is designed to provided both Shannon's principle of confusion and diffusion~\cite{shannon}. 
If $\Phi$ is an SPN, then $\Gamma_{\infty}(\Phi) = \Span{\rho, T_n}$\cite{carantiprimitive}.
\subsection{A model for the Lai-Massey scheme}
We introduce here our algebraic description of the Lai-Massey scheme~\cite{lai1990proposal} as presented by Vaudenay~\cite{vaudenay1999lai}.
\begin{definition}\label{LMdef}
Let $r$ be a non-negative integer, $\rho \in \Sym(V)\setminus \AGL(V)$ and $\theta \in \GL(V)$.
An $r$-round Lai-Massey cipher $\mathrm{LM}(\rho,\theta)$ is a set of encryption functions \[\{E_K \mid K \in \mathcal K\} \subseteq \Sym(V\times V)\]  such that  for each $K \in \mathcal K$ the map $E_K$ is the composition of $r$ functions, i.e.\
$E_K =\overline{\veps_{1,K}}\:\overline{\veps_{2,K}}\ldots\overline{\veps_{r,K}}$.
The $i$-th round function 
$\overline{\veps_{i,K}}$ is defined as
\[
\overline{\veps_{i,K}} \deq \rb\,
\tb
\sigma_{({k_i}\theta,k_i)}, 
\]
where 
\begin{itemize}
\item $\rb$ denotes the formal operator $\rhoLM \in \Sym(V\times V)$;
\item $\tb$ denotes the formal operator $\begin{pmatrix}
\theta & \mbb0\\
\theta & \mbb1
\end{pmatrix} \in \Sym(V\times V)$;
\item the  key-schedule $\mathcal{K} \rightarrow V^{r}$, $K \mapsto (k_1,k_2,\ldots,k_r)$ is surjective with respect to any round.
\end{itemize}
\end{definition}
By the assumption on the key-schedule, we can always assume without loss of generality that $0\rho = 0$, provided that,
in each round, the value $(0\rho\theta,0\rho)$ is added to the round key of the previous iteration.
It is well known that, for security concerns, the function $\theta$ is required to be an orthomorphism~\cite{vaudenay1999lai}. However,
we do not make use of this hypothesis in our analysis. We strongly use, instead, the assumption that such a function is linear.\\

The general round function of a Lai-Massey cipher is displayed in Fig.~\ref{fig.rounds}.
Notice that the previous formal definition coincides with the classical
definition given by Vaudenay~\cite{vaudenay1999lai}. 
Indeed, given $(x,y) \in V\times V$ we have 
\[(x,y)\overline{\veps_{i,K}} = \big((x+(x+y)\rho+k_i)\theta, y+(x+y)\rho+k_i\big).\]

\noindent Moreover, it is easy to check that $\overline{\veps_{i,K}}$ is invertible with the following inverse
\[
\overline{\veps_{i,K}}^{\,-1}  = \tb^{-1}\rb^{-1}\sigma_{(k_i,k_i)},
\] 
where $\rb^{-1} = \rhoinvLM$ and 
$\tb^{-1} = \thetainvLM$.

\noindent Note that, as in the Feistel Network case, the inverse $\overline{\veps_{i,K}}^{\, -1}$ of the round function $\overline{\veps_{i,K}}$ of 
a Lai-Massey cipher does not involve the inverse of $\rho$. We have nonetheless assumed that $\rho$ is bijective, since in our result it is used 
as the generator of a group.

\noindent It is worth mentioning that even if IDEA was the starting point for the definition of the LM framework, it does not fit in the presentation of Definition~\ref{LMdef}, since e.g.\ in IDEA the round key is mixed to the state by  using operations different from the XOR.\\ 

\begin{figure}
\centering
\includegraphics[scale= 1]{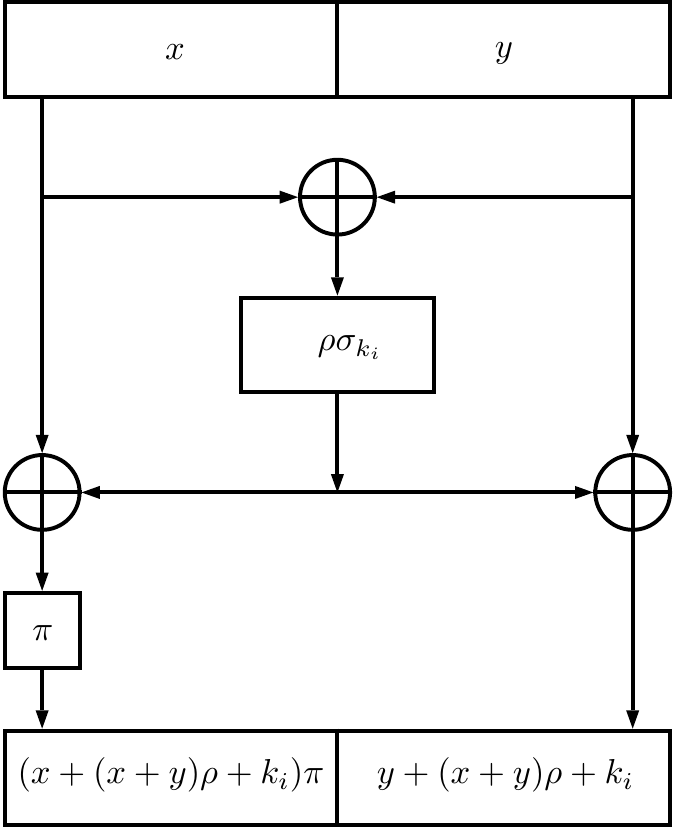}
\caption{The $i$-th round function of a Lai-Massey cipher}\label{fig.rounds}
\end{figure}

Let us define the group 
\begin{equation}\label{eqgr}
\Gamma(\mathrm{LM}(\rho,\theta)) \deq \Span{\rb, \tb, T_{2n}},
\end{equation}
which clearly contains the group $\Span{\rb\,\tb, T_{2n}}$ generated by the round functions of a Lai-Massey cipher. 
Notice that considering every possible translation in $T_{2n}$ in Eq.\eqref{eqgr} implicitly means
that the study is carried out without considering the role of the particular choice of the key-schedule.
This is however a common practice in the study of the primitivity of the groups generated by the
rounds of a cipher, except for one recent result~\cite{arafeistel}.

In the following section we will prove our main contribution, i.e.\ that $\Gamma(\mathrm{LM}(\rho,\theta))$ is primitive 
provided that $\Span{\rho, T_n}$ is primitive. It is worth mentioning again that 
$\Span{\rho, T_n}$ is the group generated by the round functions of the SPN whose  function $\rho$ is
the same that composes the building block for the round functions of the LM. In this sense, the primitivity of a 
Lai-Massey scheme is reduced to the one of the corresponding SPN.


\section{The primitivity reduction}\label{sec:pri}
\begin{definition}
A subgroup $  U \leq V\times V$ is a \emph{linear block} for $f \in \Sym(V\times V)$ if for each $(v,w) \in V\times V$ there exists $(v',w') \in V \times V$ such
that 
\begin{equation}\label{lin_block}
(U +(v,w))f = U+(v',w').
\end{equation} Notice that we can always assume $(v',w') = (v,w)f.$ 
\end{definition}
In the following result we assume the existence of a linear block $U$ for $\rb$. In this case we have
\[(v',w') = (v+w, w+(v+w)\rho) 
.\]
 Moreover, it is easy to check that 
$U$ is a linear block also for $\rb^{-1}$, from which we  obtain 
\[(U +(v,w))\rb^{-1} = U +(v+w+v\rho,w+v\rho).\]
We use Theorem~\ref{gours} and Corollary~\ref{lemma:psiforphi} to 
provide an useful decomposition of $U$. The explicit dependence of all the groups
from $\rho$  is here omitted.
\begin{lemma}\label{cond}
Let $U \leq V \times V$, and let $A,B,C,D \leq V$ and $\phi: A \to C$ an homomorphism such that
$U = \{ (a, a \varphi + d) \mid a \in A, d \in D\}$.
Let us assume that $U$ is a linear block for $\rb$.
Then the following conditions hold: 
\begin{enumerate}
\item \label{item_DinA} $D \leq A$;
\item \label{item_AphiinA} $A \phi \leq A$;
\item \label{item_DinDphi} $D \phi \leq D$.
\end{enumerate}
\end{lemma}
\begin{proof}
Since $U$ is a linear block for $\rb$, taking $u=v=0$ in Eq.~\eqref{lin_block} we have 
that for each $a \in A$ and $d \in D$ there exist $x \in A$ and $y \in D$ such that 

\begin{eqnarray*}
(a,a\phi+d)\rb &=& (a,a\phi+d)\rhoLM \\
&=& (a+a\phi+d,a\phi+d + (a+a\phi+d)\rho)\\ &=& (x,x\phi+y).
\end{eqnarray*}
 If $a=0$, then 
$x=d$ and therefore $D \leq A$, so (\ref{item_DinA}) is proved.
If $d=0$, then $x = a + a\phi$ and therefore $A\phi \leq A$, which is (\ref{item_AphiinA}).
As noticed, $U$ is also a linear block for $\rb^{-1}$, hence for each $a \in A$ and $d \in D$ there exist $x \in A$ and $y \in D$ such that
\begin{eqnarray*}
(a,a\phi+d) \rb^{-1} &=&(a,a\phi+d)\rhoinvLM\\
&=& (a+a\phi+d+a\rho, a\phi+d+a\rho)\\ &=& (x,x\phi + y).
\end{eqnarray*}
If $a=0$, then $y=d\phi+d$, and consequently $d\phi \in D$, which proves (\ref{item_DinDphi}). 
\end{proof}
We now use the previous lemma to show our main result on the primitivity of the Lai-Massey scheme. Notice that the result is valid for any choice of $\theta$.
\begin{theorem}\label{thm:main}
If $\Span{\rho, T_n}$ is primitive, then $\Gamma_{\infty}(\mathrm{LM}(\rho,\theta))$ is primitive.
\end{theorem}
\begin{proof}
It is enough to prove that $\Span{\rb, T_{2n}}$ is primitive. Let us assume that it is imprimitive, i.e.\ that there exists
a block system $\mathcal U$ for $\Span{\rb, T_{2n}}$. Then, from Lemma~\ref{translatioBlocks},  the block system 
is
 $\mathcal U = \{U + (v,w) \mid (v,w) \in V\times V\}$  for a non-trivial proper subspace $U$ of $V\times V$.
 Since $U$ is a linear block for $\rb$, we have that for each $(v,w) \in V\times V$ and for each $a \in A$ and $d \in D$ there exist
  $x \in A$ and $y \in D$ such that 
\begin{equation*}
\begin{split}
  (a+v, a\phi + d+ w)\rb  & =\\
  &  = (a+a\phi+d+v+w, (a+a\phi+d+v+w)\rho + a\phi+d+w) \\
 &= (x+v+w, x\phi+y+w+(v+w)\rho)
\end{split}
\end{equation*}
  If $a=0$, then $x=d$ and 
  \[
  y+ d\phi +d+(v+w)\rho= (d+v+w)\rho.
  \]
From Lemma~\ref{cond} we have $d\phi \in D$, and therefore, since $\rho$ is bijective, we obtain the equality
\[
(D+v+w) \rho = D+(v+w) \rho.
\]
If $D$ is a non-trivial proper subgroup of $V$, then $\{D+v \mid v \in V\}$ is a block system for $\Span{\rho, T_n}$, which proves our claim.
In order to conclude the proof, let us prove that both the assumptions $D=\{0\}$ and $D=\F_2^n$ lead to contradictions.\\
{$\mathbf{[D = \F_2^n]}$} Since $D \leq A$, then $A= \F_2^n$, and therefore $B=C = \F_2^n$, since from the hypothesis $A/B \cong C/D$. This proves that $U$ is not proper, a contradiction.\\
{$\mathbf{[D = \{0\}]}$} In this case we have $U = \{(a,a\phi) \mid a \in A\}$, and so, for each $(v,w) \in V\times V$ and for each $a \in A$ there
exists $x \in A$ such that 
\begin{eqnarray}
(a+v,a\phi+w)\rb &=& (a+a\phi+v+w, (a+a\phi+v+w)\rho + a\phi+w) \nonumber\\
 &=& (x+v+w, x\phi+w+(v+w)\rho)\nonumber,
\end{eqnarray}
then $x= a+a\phi$ and 
\begin{equation}\label{eqth}
 (a+a\phi+v+w)\rho = a\phi^2+(v+w)\rho
\end{equation}
Since $B\phi \leq D$, then $B\phi = \{0\}$ and so, if $a \in B$ from Eq.~\eqref{eqth} we obtain 
\[
 (a+v+w)\rho =(v+w)\rho,
\]
which implies $a= 0$, i.e.\ $B=\{0\}$. This proves that $\phi = \psi : A \rightarrow C$ is an isomorphism.
But $A\phi \leq A$, from Lemma~\ref{cond},  therefore  $	\phi$ is an automorphism of $A$ and,
from Eq.\eqref{eqth}, we obtain 
\[
(A+v+w)\rho = A + (v+w)\rho.
\]
In the case under consideration, i.e.\ when $D = \{0\}$, the claim is proved by showing that $\{A+v \mid v \in V\}$ is a block system. This is 
addressed in the remainder of the proof.
Let us prove that $A$ is non-trivial and proper. If $A = \{0\}$, then $C = D = \{0\}$, and so also $B = \{0\}$, therefore $U$ is trivial, a contradiction. To conclude, let us assume $A = \F_2^n$. From Eq.~\eqref{eqth}, setting $v=w=0$, we obtain that $a\phi^2=(a+a\phi)\rho$. If $a \in A$ is a fixed point of $\phi$, i.e.\ $a = a\phi$, then $(a+a\phi)\rho = 0$, and so $a\phi^2 = 0$. Therefore $a = 0$, since $\phi$ is an automorphism. We have proved that $\phi$ is fixed-point free, except for the trivial one $a=0$, from which it follows that  $\mbb1 + \phi$ is injective and, since $A \phi \leq A$, we have 
$\{a+a\phi \mid a \in A\} = A = \F_2^n$. Therefore $\rho$ is linear on $\F_2^n$, a contradiction.
\end{proof} 
We have already observed that $\Span{\rho, T_n}$ is the group $\Gamma_{\infty}(\Phi)$ generated by the rounds of the Substitution Permutation Network $\Phi$ whose  
$i$-th round function is  $\veps_{i,K} = \rho\sigma_{k_i}$ for some $\rho = \gamma\lambda \in \Sym(V)\setminus \AGL(V)$. The conditions which prove $\Gamma_{\infty}(\Phi)$  primitive
in the case of the SPNs has been extensively studied, due to the popularity of the design framework. It has been proved that the primitivity is granted when the confusion 
layer $\gamma$ of $\rho$ satisfies some well-established
conditions of non-linearity, provided that $\lambda$ provides sufficient diffusion. The interested reader 
may refer to Caranti et al.~\cite{carantiprimitive}, where the primitivity of SPNs is studied in the larger context of \emph{translation-based} ciphers.
\subsection*{Different rounds}
In our analysis we have assumed that, for sake of simplicity, in both the cases of SPNs and of LMs the 	\emph{same} round function is applied to each round,
with the only exception of the round key. It is worth being mentioned here that almost no real cipher can exactly fit this model, due to the
natural need to differentiate the encryption routine by introducing some atypical rounds, for both security and efficiency reasons (see e.g.\ the first and 
last rounds of AES~\cite{aes}). However, this does not represent an actual limitation when it comes to evaluate the security of a design 
from the point of view of the group generated by the round function. If we assume, indeed, that our target SPN and LM feature different round functions 
for each round, then the primitivity of $\Span{\rb_i, \tb_i, T_{2n}}$ can be reduced to that of $\Span{\rho_i, T}$, proceeding as in~Theorem~\ref{thm:main}.
When this is true for one round $i$, then the full group
\begin{equation}\label{eqgr2}
\Span{\Span{\rb_i, \tb_i, T_{2n}}\mid 1 \leq i \leq r}
\end{equation} is primitive.

Assuming that the definition of SPN and Definition~\ref{LMdef} allow different round functions for each round, then the following consequence of Theorem~\ref{thm:main}
can be derived.
\begin{corollary}
If a given round of an SPN generates a primitive group, then the group defined as in Eq.~\eqref{eqgr2} generated by the rounds of the corresponding LM is primitive.
\end{corollary}

\subsection*{2-transitivity}
It is well known that every 2-transitive group is primitive. It may be natural to ask whether an alternative version of Theorem~\ref{thm:main}
could be obtained, where primitivity is replaced by 2-transitivity. In the case $n=3$, we have exhaustively searched using \texttt{Magma}~\cite{magma} for  all the non-linear functions $\rho$
such that $\Span{\rho, T_3}$ is a 2-transitive group.  For all of those, $\Span{\rb, T_6}$ is always 2-transitive. Setting $n=4$, a partial search in the space led to no counterexamples.
A brute-force search when $n\geq 4$ is out of the scopes of this work since requires code optimization and a faster programming language.
On the other side, it is well known that the $2$-transitivity of $\Gamma_{\infty}(\mathrm{LM}(\rho,\theta))$ is equivalent to the transitivity of the stabilizer of $(0,0)$
 on $V \times V\setminus\{(0,0)\}$. However, a description of $\Gamma_{\infty}(\mathrm{LM}(\rho,\theta))_{(0,0)}$ is not 
easily obtained from $\Span{\rho, T_n}_0$ due to the non-linear dependence introduced by the Lai-Massey formal operator.
For this reasons, at the time of writing we are not able to conjecture that the 2-transitivity of $\Span{\rho, T_n}$ implies
that $\Gamma_{\infty}(\mathrm{LM}(\rho,\theta))$ is 2-transitive, and we leave this as an open problem.

\section*{Acknowledgment}
The authors are thankful to Massimiliano Sala and Ralph Wernsdorf for some useful discussions and suggestions.
\bibliographystyle{amsalpha}
\bibliography{sym2n_ref}

\newcommand{\etalchar}[1]{$^{#1}$}
\providecommand{\bysame}{\leavevmode\hbox to3em{\hrulefill}\thinspace}
\providecommand{\MR}{\relax\ifhmode\unskip\space\fi MR }
\providecommand{\MRhref}[2]{%
  \href{http://www.ams.org/mathscinet-getitem?mr=#1}{#2}
}
\providecommand{\href}[2]{#2}
\begin{thebibliography}{ACC{\etalchar{+}}19}

\bibitem[ACC{\etalchar{+}}19]{aragonawave}
Riccardo Aragona, Marco Calderini, Roberto Civino, Massimiliano Sala, and
  Ilaria Zappatore, \emph{Wave-shaped round functions and primitive groups},
  Adv. Math. Commun. \textbf{13} (2019), no.~1, 67--88.

\bibitem[ACC20]{arafeistel}
Riccardo Aragona, Marco Calderini, and Roberto Civino, \emph{Some
  group-theoretical results on {F}eistel networks in a long-key scenario}, Adv.
  Math. Commun. \textbf{14} (2020), no.~4, 727--743.

\bibitem[ACS17]{aragona2017group}
Riccardo Aragona, Aragona Caranti, and Massimiliano Sala, \emph{The group
  generated by the round functions of a {GOST}-like cipher}, Ann. Mat. Pura
  Appl. (4) \textbf{196} (2017), no.~1, 1--17.

\bibitem[ACTT18]{aragona2018primitivity}
Riccardo Aragona, Marco Calderini, Antonio Tortora, and Maria Tota,
  \emph{Primitivity of {PRESENT} and other lightweight ciphers}, J. Algebra
  Appl. \textbf{17} (2018), no.~6, 1850115, 16.

\bibitem[Cam99]{cameronpermutation}
Peter~J. Cameron, \emph{Permutation groups}, London Mathematical Society
  Student Texts, vol.~45, Cambridge University Press, Cambridge, 1999.

\bibitem[CCS17]{calderini2017translation}
Marco Calderini, Roberto Civino, and Massimiliano Sala, \emph{Translation
  groups of the affine general linear group over prime fields}, arXiv preprint
  arXiv:1702.00581 (2017).

\bibitem[CDVS09]{carantiprimitive}
Andrea Caranti, Francesca Dalla~Volta, and Massimiliano Sala, \emph{On some
  block ciphers and imprimitive groups}, Appl. Algebra Engrg. Comm. Comput.
  \textbf{20} (2009), no.~5-6, 339--350.

\bibitem[CG75]{copgross}
Don Coppersmith and Edna Grossman, \emph{Generators for certain alternating
  groups with applications to cryptography}, SIAM J. Appl. Math. \textbf{29}
  (1975), no.~4, 624--627.

\bibitem[CP96]{magma}
John Cannon and Catherine Playoust, \emph{M{AGMA}: a new computer algebra
  system}, Euromath Bull. \textbf{2} (1996), no.~1, 113--144.

\bibitem[DR02]{aes}
Joan Daemen and Vincent Rijmen, \emph{The design of {R}ijndael}, Information
  Security and Cryptography, Springer-Verlag, Berlin, 2002, {A}{E}{S} --- the
  {A}dvanced {E}ncryption {S}tandard.

\bibitem[Fei73]{feistel}
Horst Feistel, \emph{{C}ryptography and {C}omputer {P}rivacy}, Scientific
  American \textbf{228} (1973), no.~5, 15--23.

\bibitem[GJ14]{impossibledifferentialLM}
Rui Guo and Chenhui Jin, \emph{Impossible differential cryptanalysis on
  {L}ai-{M}assey scheme}, ETRI Journal \textbf{36} (2014), no.~6, 1032--1040.

\bibitem[Gou89]{goursat}
Edouard Goursat, \emph{Sur les substitutions orthogonales et les divisions
  r\'{e}guli\`eres de l'espace}, Ann. Sci. \'{E}cole Norm. Sup. (3) \textbf{6}
  (1889), 9--102.

\bibitem[KRS88]{kaliski}
Burton~S. Kaliski, Jr., Ronald~L. Rivest, and Alan~T. Sherman, \emph{Is the
  {D}ata {E}ncryption {S}tandard a group? ({R}esults of cycling experiments on
  {DES})}, J. Cryptology \textbf{1} (1988), no.~1, 3--36.

\bibitem[LLH15]{randomlaimassey}
Yiyuan Luo, Xuejia Lai, and Jing Hu, \emph{The pseudorandomness of many-round
  {L}ai-{M}assey scheme}, JISE J. Inf. Sci. Eng. \textbf{31} (2015), no.~3,
  1085--1096.

\bibitem[LLZ17]{genericattackLM}
Yiyuan Luo, Xuejia Lai, and Yujie Zhou, \emph{Generic attacks on the
  {L}ai-{M}assey scheme}, Des. Codes Cryptogr. \textbf{83} (2017), no.~2,
  407--423.

\bibitem[LM91]{lai1990proposal}
Xuejia Lai and James~L. Massey, \emph{A proposal for a new block encryption
  standard}, Advances in cryptology---{EUROCRYPT} '90 ({A}arhus, 1990), Lecture
  Notes in Comput. Sci., vol. 473, Springer, Berlin, 1991, pp.~389--404.

\bibitem[Pat99]{paterson1999imprimitive}
Kenneth~G. Paterson, \emph{{I}mprimitive {P}ermutation {G}roups and {T}rapdoors
  in {I}terated {B}lock {C}iphers}, Fast Software Encryption (Berlin,
  Heidelberg), Springer Berlin Heidelberg, 1999, pp.~201--214.

\bibitem[Sha49]{shannon}
Claude~E. Shannon, \emph{Communication theory of secrecy systems}, Bell System
  Tech. J. \textbf{28} (1949), 656--715.

\bibitem[SW08]{sparwenrij}
R\"{u}diger Sparr and Ralph Wernsdorf, \emph{Group theoretic properties of
  {R}ijndael-like ciphers}, Discrete Appl. Math. \textbf{156} (2008), no.~16,
  3139--3149.

\bibitem[SW15]{wenkas}
\bysame, \emph{The round functions of {KASUMI} generate the alternating group},
  J. Math. Cryptol. \textbf{9} (2015), no.~1, 23--32.

\bibitem[Vau99]{vaudenay1999lai}
Serge Vaudenay, \emph{On the {L}ai-{M}assey scheme}, Advances in
  cryptology---{ASIACRYPT}'99 ({S}ingapore), Lecture Notes in Comput. Sci.,
  vol. 1716, Springer, Berlin, 1999, pp.~8--19.

\bibitem[Wer01]{wenidea}
Ralph Wernsdorf, \emph{{IDEA}, {SAFER}++ and {T}heir {P}ermutation {G}roups,},
  Second Open NESSIE Workshop, Royal Holloway University of London, Egham,
  2001.

\bibitem[YPL11]{quasifeistel}
Aaram Yun, Je~Hong Park, and Jooyoung Lee, \emph{On {L}ai-{M}assey and
  quasi-{F}eistel ciphers}, Des. Codes Cryptogr. \textbf{58} (2011), no.~1,
  45--72.

\end{thebibliography}

\end{document}